\documentclass[11pt]{llncs}
\usepackage{graphicx}
\usepackage{url}
\usepackage{amsmath}
\usepackage{amssymb}
\usepackage{amsfonts}
\usepackage{subfigure}
\usepackage{enumerate}
\usepackage{paralist}
\usepackage{tikz}
\usepackage[textsize=footnotesize]{todonotes}

\usepackage{amsthm}
\usepackage{cite}
\usepackage{color}

\newcommand\snip[1]{}
\newcommand{\indeg}{\textrm{indeg}}
\newcommand{\outdeg}{\textrm{outdeg}}

\newcommand{\leftw}{\textrm{left}}
\newcommand{\rightw}{\textrm{right}}

\renewcommand{\leq}{\leqslant} 
\renewcommand{\geq}{\geqslant}

\newcommand{\po}{\preceq}

\newcommand{\tightoverset}[2]{%
  \mathop{#2}\limits^{\vbox to -.5ex{\kern-0.6ex\hbox{$#1$}\vss}}}
\newcommand{\rhup}[1]{\tightoverset{\rightharpoonup}{#1}}

\makeatletter
\def\cramped                           
 {
\parskip0pt\@topsep0pt       
  \itemsep0pt\parsep0pt
  \settowidth{\labelwidth}{\@itemlabel}
  \advance\leftmargin-\labelsep
  \advance\leftmargin-\labelwidth
  \parshape1\@totalleftmargin\linewidth
}
\makeatother

\title{Bar 1-Visibility Graphs and their relation to other Nearly Planar Graphs\thanks{%
The research reported in this paper started at the 2013
McGill/INRIA/UVictoria Bellairs workshop. We gratefully acknowledge
discussions with the other participants.
Research supported by NSERC, and  by MIUR of Italy
under project AlgoDEEP prot. 2008TFBWL4.
}}
\author{
W. Evans\inst{1} \and
M. Kaufmann \inst{2} \and
W. Lenhart \inst{3}\and
G. Liotta\inst{4} \and
T. Mchedlidze\inst{5}\and
S. Wismath\inst{6}
}
\institute{
University of British Columbia, Canada
\and
Universit{\"a}t T{\"u}bingen, Germany
\and
Williams University, U.S.A.
\and
Universit{\'a} degli Studi di Perugia, Italy
\and
Karlsruhe Institute of Technology, Germany
\and
University of Lethbridge, Canada
}

\begin{document}

\pagenumbering{arabic}
\pagestyle{plain}

\maketitle

\begin{abstract}
A graph is called a \emph{strong (\emph{resp.} weak) bar 1-visibility graph} 
if its vertices can be represented as horizontal segments (bars) in the
plane so that its edges are all (resp. a subset of) the pairs of
vertices whose bars have a $\epsilon$-thick vertical line connecting
them that intersects at most one other bar.

We explore the relation among weak (resp. strong) bar 1-visibility graphs and other 
nearly planar graph classes.
In particular, we study their relation to
\emph{1-planar graphs},
which have a drawing with at most one crossing per edge;
\emph{quasi-planar graphs},
which have a drawing with no three mutually crossing edges;
the squares of \emph{planar 1-flow networks},
which are upward digraphs with in- or out-degree at most one.
Our main results are
that 1-planar graphs and the (undirected) squares of planar 1-flow
networks are weak bar 1-visibility graphs and that these are
quasi-planar graphs.
\end{abstract}

\section{Introduction}

Developing a theory of graph drawing beyond planarity has received
increasing interest in recent years. This is partly motivated by
applications of network visualization, where it is important to
compute readable drawings of non-planar graphs. Within this research
framework, a rich body of papers has in particular been devoted to
the study of the combinatorial properties of different types of
drawings that are nearly planar, i.e., do not allow a specific restricted set of crossing configurations, such as the crossings cannot form too sharp angles (see, e.g.,~\cite{dl-cargd-12} for a survey).
Another study of visualizations of non-planar graphs that are ``close to planar'' was conducted by Dean et al.~\cite{degl7}, by introducing so-called 
bar $k$-visibility graphs and representations. Dean et al. were particularly interested in measurements of closeness to planarity of
bar $k$-visibility graphs. In this work we shed some light on this question by investigating the relation of
bar $1$-visibility graphs with graphs that are known to be ``close to planar''.  Thus, we study the relation of bar $1$-visibility graphs
with nearly planar graphs, particularly 1-planar and quasi-planar graphs. Moreover, we investigate the  relation of bar $1$-visibility graphs with squares of planar graphs.


A \emph{bar layout} consists of $n$ horizontal
non-intersecting line segments (bars).  A pair of bars $u$ and $v$
are \emph{$k$-visible} if and only if there is an 
axis-aligned rectangle of non-zero width touching $u$ and $v$ which
intersects at most $k$ bars in the layout.
For a given bar layout, its (unique) \emph{strong bar $k$-visibility
graph} has a vertex for every bar and an edge $(u,v)$ if and only if
the corresponding bars $u$ and $v$ are $k$-visible.
A \emph{weak bar $k$-visibility graph} of a bar layout is any
(spanning) subgraph of its strong bar $k$-visibility graph.
Note that there are $2^m$ weak bar $k$-visibility graphs if
there are $m$ edges in the strong bar $k$-visibility graph.
A graph is a strong (weak) bar $k$-visibility graph if it is the
strong (weak) bar $k$-visibility graph of some bar layout.
Independently, Wismath~\cite{Wismath85} and Tamassia and
Tollis~\cite{tt86} characterized strong bar $0$-visibility graphs as
exactly those that have a planar embedding with all cut vertices on
the exterior face. Weak bar $0$-visibility graphs are exactly the
planar graphs~\cite{dett-gdavg-99}. Dean et al.~\cite{degl7} showed
that $K_n$ $(n \leq 8)$ is a strong bar $1$-visibility graph, that
$K_9$ is not a strong bar $1$-visibility graph, and that all
$n$-vertex strong (and thus weak) bar $1$-visibility graphs have fewer
than $6n-20$ edges.
Felsner and Massow~\cite{DBLP:journals/jgaa/FelsnerM08}
showed that there exists a strong bar $1$-visibility graph that has
thickness three, disproving an earlier conjecture~\cite{degl7} that
all such graphs have thickness two or less.

While bar layouts represent the vertices of a graph as horizontal
segments, a {\em topological drawing} of a graph $G$ maps each
vertex $u$ of $G$ to a distinct point $p_u$ in the plane, each edge
$(u,v)$ of $G$ to a Jordan arc connecting $p_u$ and $p_v$ and not
passing through any other vertex, and is such that any two edges
have at most one point in common. A \emph{$k$-planar graph} is one
which admits a topological drawing in which each edge is crossed by
at most $k$ other edges. Pach and T\'oth proved that $1$-planar
graphs with $n$ vertices have at most $4n-8$ edges, which is a
tight upper bound~\cite{PachT97} and
that, in general, $k$-planar graphs are sparse. Korzhik and Mohar
proved that recognizing $1$-planar graphs is
NP-hard~\cite{DBLP:conf/gd/KorzhikM08}. A limited list of additional
papers on $k$-planar graphs
includes~\cite{DBLP:conf/gd/AuerBGH12,DBLP:journals/dam/BorodinKRS01,DBLP:conf/gd/BrandenburgEGGHR12,DBLP:journals/dam/CzapH12,DBLP:journals/dam/EadesL13,DBLP:conf/gd/EadesHKLSS12,DBLP:journals/dm/FabriciM07,DBLP:conf/cocoon/HongELP12,DBLP:journals/dm/Korzhik08a,DBLP:journals/dm/Suzuki10,th-rdg-88}. The relation between $1$-planar and bar $1$-visibility graphs was recently investigated in~\cite{DBLP:journals/corr/abs-1302-4870,SultanaRRT13}, where it was proven that several restricted subclasses of $1$-planar graphs are weak bar $1$-visibility graphs.

A \emph{$k$-quasi-planar graph} admits a topological drawing such that
no $k$ edges mutually cross; $3$-quasi-planar graphs are commonly
called \emph{quasi-planar}, for short. Ackerman and Tardos showed that
quasi-planar graphs with $n$ vertices have at most $6.5 n - O(1)$
edges~\cite{Ackerman2007563}. Giacomo et
al.~\cite{DBLP:conf/wg/GiacomoDLM12} described how to construct linear
area $k$-quasi-planar drawings of graphs with bounded treewidth.
Recently, Geneson et al.~\cite{GenesonKT13} showed that all semi-bar
$k$-visibility graphs\footnote{Semi-bar visibility graphs require all horizontal bars to have minimum
$x$-coordinate equal to zero~\cite{DBLP:journals/jgaa/FelsnerM08}.}
are $(k+2)$-quasi-planar.
See also~\cite{AckermanE09,pachss96} for additional references about
$k$-quasi-planar graphs.

Another family of non-planar graphs, which are in some sense ``close to planar'' are the squares of
directed planar graphs with bounded in- or out- degree. The
\emph{square} $G^2$ of a graph $G = (V,E)$ has vertex set $V$ and all
edges $(u,v)$ where there is a path of length at most two from $u$ to
$v$ in $G$. Observe that if for each vertex of a directed planar graph
$G$, either in- or out- degree is bounded by a constant, then the
number of edges in $G^2$ is linear.
This fact is captured by the notion of $k$-flow networks. A
(\emph{planar}) \emph{$k$-flow network} is a (upward planar)
directed graph in which every vertex $v$ has $\min\{\indeg(v),
\outdeg(v)\} \leq k$. The name of the class stems from the fact that
at most $k$ units of flow can pass through each
vertex. Tarjan~\cite{tarjan1983data} studied $1$-flow networks under
the name of \emph{unit flow networks}. Bessy et al.~\cite{BessyHB06}
studied the arc-chromatic number of $k$-flow networks under the name of
\emph{$(k \vee k)$-digraphs}.
We let $k$-flow$^2$ denote the class of graphs that are the squares of
planar $k$-flow networks.  Squares of graphs arise naturally in
understanding bar $1$-visibility graphs since a bar layout that
represents a bar $0$-visibility graph $G$ also represents a family of
weak bar $1$-visibility graphs each of which is a spanning subgraph of
$G^2$. That is, every weak bar $1$-visibility graph is a spanning
subgraph of the square of a bar $0$-visibility graph. Thus, it is
natural to consider which bar $0$-visibility graphs have squares that
are weak bar $1$-visible.

While several properties of bar 1-visibility graphs have been investigated, it 
remains an open problem to provide their complete characterization.
Recall that bar 1-visibility graphs are generally non-planar and contain at most
$6n-20$ edges. Observe that this number is greater than the maximum number of edges in 
1-planar graphs (at most $4n-8$) and smaller than the maximum number of edges in quasi-planar graphs (at most  $6.5 n - O(1)$).
Recall also that, every weak bar $1$-visibility graph is a spanning
subgraph of the square of a bar $0$-visibility graph. Motivated by these facts we study the relation
of bar 1-visibility graphs with families of 1-planar, quasi-planar and squares of planar graphs.
Our contribution is threefold:
(i) We show that the class of weak bar $1$-visibility graphs contains
the class of $1$-planar graphs, which proves a conjecture of Sultana,
Rahman, Roy, and
Tairin~\cite{DBLP:journals/corr/abs-1302-4870,SultanaRRT13},
(ii) We show that the class of bar $1$-visibility graphs is contained
in the class of quasi-planar graphs, and
(iii) We show that $1$-flow$^2$ graphs are weak bar $1$-visibility
graphs, and that this is not always true for $2$-flow$^2$ graphs. An
overview of our results is illustrated in Figure~\ref{fig:world} and
thoroughly described in Section~\ref{se:overview}. Proof details about
the inclusion relationships of Figure~\ref{fig:world} are given in
Sections~\ref{se:1P_is_WeB1}, \ref{se:WeB1_is_QP}, and
\ref{se:F1sq_is_WeB1}.

We notice that proof of (i) was recently independently obtained by Brandenburg~\cite{Brandenburg13}.


\begin{figure}
\begin{center}
\includegraphics{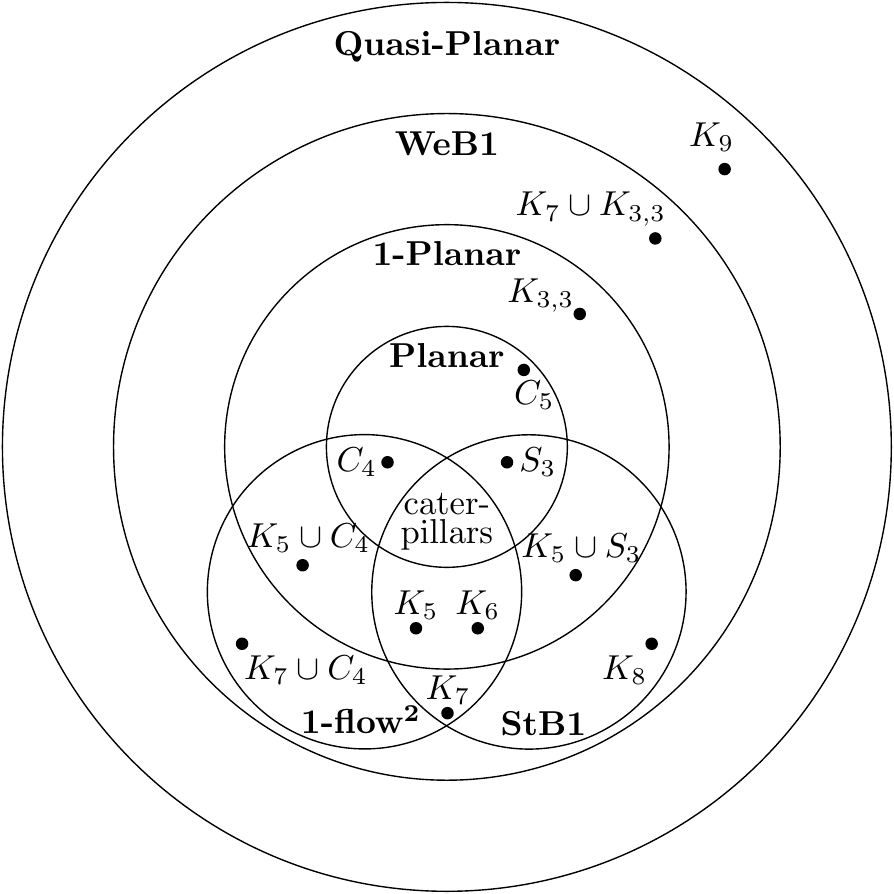}
\end{center}
\caption{Relationships among graph classes proved in this paper.}
\label{fig:world}
\end{figure}


\section{Graph classes and their relationships}\label{se:overview}

In this section we describe Figure~\ref{fig:world}. We abbreviate
strong and weak bar $1$-visibility graphs as \emph{StB1} and
\emph{WeB1} graphs. Since a strong bar $1$-visibility graph is a
weak bar $1$-visibility graph of the same bar layout, it follows
that $StB1 \subseteq WeB1$. The observation that every planar graph
is WeB1 (it is in fact a weak bar $0$-visibility
graph~\cite{dett-gdavg-99}); the fact that $K_{3,3}$ is WeB1 $\Big(
\vcenter{\hbox{\includegraphics{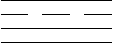}}} \Big)$;
and the following
simple lemma prove that $StB1 \subset WeB1$.

\begin{lemma}
\label{le:Knm}
Any graph that is StB1 is either a forest or contains a triangle.
\end{lemma}
\begin{proof}
Let $G$ be StB1 and suppose $G$ contains a cycle but not a triangle.
In the strong bar 1-visibility layout, let $v$ be a vertex in a cycle
whose bar has right endpoint with
minimum $x$-coordinate, $x$. Since $v$ has at least two neighbors that
are in a cycle,
their bars must share some $x$-coordinate with bar $v$ and all must
span $x$. Thus at least three bars span $x$ implying a
triangle in the graph, which is a contradiction.
\end{proof}


The number of edges in any 1-planar graph is known to be at most
$4n-8$~\cite{PachT97}. Thus, $K_7$ and $K_8$ are not 1-planar (too
many edges) but are StB1 as proved by Dean et al.~\cite{degl7}.
The disjoint union $K_7 \cup K_{3,3}$ is WeB1 but it is not 1-planar (because of $K_7$) and
it is not StB1 (because of $K_{3,3}$ by Lemma~\ref{le:Knm}). We
 show that all 1-planar graphs are
WeB1 (see Section~\ref{se:1P_is_WeB1}) and that all WeB1 graphs are
quasi-planar (see Section~\ref{se:WeB1_is_QP}).

In Section~\ref{se:F1sq_is_WeB1}, we show that
$1$-flow$^2$ graphs are WeB1. We also show that $2$-flow$^2$ graphs
are not always WeB1.
%
It is easy to see that if $G^2 \neq G$ then $G^2$ contains a triangle.
Thus, since $K_{3,3}$ is not planar and
does not contain a triangle, it is not a $1$-flow$^2$ graph.
However, every planar bipartite graph $G$ can be directed (from
one bipartition to the other) so that $G$ is a $1$-flow network with
$G^2 = G$ and is thus a $1$-flow$^2$ graph. Therefore, caterpillars and $C_4$ are
$1$-flow$^2$ graphs. It is also easy to see that caterpillars are StB1.
Let $G$ is the 1-flow graph of Figure~\ref{fig:square_is_k7}, then the square of the subgraph of $G$ induced 
by vertices $1,\dots,n$  is $K_n$ ($n \leq 7$). In Section~\ref{se:F1sq_is_WeB1} we show that $K_8$ is not the square of a $1$-flow network, and that there exists a planar StB1 graph ($S_3$) that is not the square of a
$1$-flow network.

\begin{figure}
\begin{center}
\includegraphics{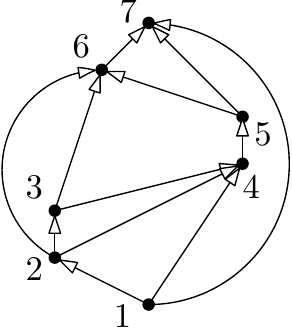}
\end{center}
\caption{$1$-flow graph $G$ such that the square of the subgraph of $G$ induced 
by vertices $1,\dots,n$  is $K_n$ ($n \leq 7$).}
\label{fig:square_is_k7}
\end{figure}

\section{1-planar graphs are WeB1}
\label{se:1P_is_WeB1}

\begin{theorem}
If a graph $G$ is 1-planar then $G$ is WeB1.
\label{thm:oneplanar}
\end{theorem} 



\begin{proof}
It suffices to prove the theorem for a maximal 1-planar graph $G=(V,E)$ since
a WeB1-representation of $G$ is a
WeB1-representation of every graph $(V,E')$ with $E' \subseteq E$.
Let $\Gamma$ be a 1-planar drawing of $G$.
Let $ab$ and $cd$ be a pair of edges that cross in $\Gamma$.
Since $G$ is a maximal 1-planar graph, $G$ contains the edges $ac$,
$cb$, $bd$, and $da$; and these edges are uncrossed in $\Gamma$.
If, for example, $G$ did not contain the edge $ac$ or $ac$ was crossed,
we could \emph{re-route} $ac$ in $\Gamma$ without introducing crossings by
following edge $ab$ from $a$ to its intersection with $cd$ and then
following $cd$ to $c$; always following slightly to $c$-side of $ab$
and the $a$-side of $cd$.



Since $G$ is a maximal 1-planar graph, the planar graph $G_0$ obtained by
removing all crossing edges from $G$ is
biconnected~\cite{DBLP:conf/gd/EadesHKLSS12} and thus
has an $st$-orientation~\cite{lec-aptg-67},
which is a partial order, $\po$, on the vertices $V$ with a single
source (minimal vertex) and a single sink (maximal vertex).
We direct the edges of $G_0$ to be consistent with this partial order;
so $uv$ is directed as $\rhup{uv}$ if $u \po v$.
Let $\rhup{G_0}$ be the directed version of $G_0$, and let $\Gamma_0$ 
be the drawing $\Gamma$ restricted to $G_0$.

For every crossing pair of edges $ab$ and $cd$ in $G$, the
(undirected) cycle $C=acbda$ exists in $G_0$ since none of its edges
are crossed in $\Gamma$.
We claim that the oriented version, $\rhup{C}$, of
$C$ consists of two directed paths with common origin and common
destination.
This claim is a slight generalization of:
\begin{lemma}[Lemma~4.1~\cite{dett-gdavg-99}]
Each face $f$ of $\rhup{G_0}$ consists of two directed
paths with common origin and common destination.
\end{lemma}
In our case, $\rhup{C}$ may not be a face of $\rhup{G_0}$; it may
contain vertices and edges.
However, if our claim is violated, we can re-route the edges of the
cycle $C$ (as above) so that $\rhup{C}$ is a face of $\rhup{G_0}$ and
contradict the previous lemma.
Thus the claim holds and there must be two consecutive edges in $C$
that are oriented in the same direction, say $\rhup{ac}$ and
$\rhup{cb}$.
See for example Fig.~\ref{fig:orientCross}(a).



\begin{figure}
\centering
\includegraphics{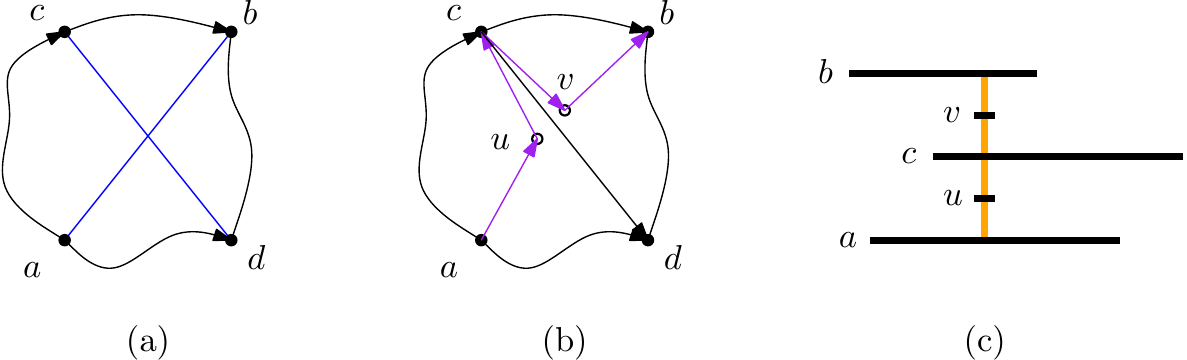}
\caption{ (a) At least two edges ($ac$ and $cb$) are oriented in the same
  direction around the cycle $C$.
(b) One edge ($ab$) in a pair of crossing edges is replaced
  with the path $aucvb$ by adding dummy
  vertices $u$ and $v$.
(c) The visibility edges of the path $aucvb$ are vertically
aligned. (Only these bars are shown.)
}
\label{fig:orientCross}
\end{figure}



We return the edge $cd$ to the drawing $\Gamma_0$ and direct it to be
consistent with the partial order, $\po$, defined by the
$st$-orientation.
In place of the edge $ab$, we insert the directed path
$aucvb$ that contains two dummy vertices, $u$ and $v$
(specifically for this crossing).
Note that, by the above discussion, this path is also consistent with
the partial order.
The dummy vertices are
placed near the point $x$ where $ab$ intersected $cd$, with edge $au$
following the drawing of $ab$ from $a$ to (near) $x$,
edge $uc$ slightly to the $a$-side of $cd$,
edge $cv$ slightly to the $b$-side of $cd$,
and edge $vb$ following the drawing of $ab$ from (near) $x$ to $b$.
See Fig.~\ref{fig:orientCross}(b).
Thus no new edge creates a crossing and the result, after every
pair of crossing edges is replaced in this fashion, is an
$st$-oriented plane graph $G'$ with drawing $\Gamma'$.
Since $G'$ is planar and has an $st$-orientation, $G'$ has a bar
0-visibility representation~\cite{Wismath85,tt86}.

The set of inserted paths are
\emph{nonintersecting}, meaning they are
edge disjoint and do not \emph{cross} at common vertices\footnote{%
Two paths \emph{cross} at a vertex $v$ in a drawing $\Gamma$ if $v$
has four incident edges $e_1$, $e_2$, $e_3$, and $e_4$ in clockwise
order such that one path contains $e_1$ and $e_3$ while the other
contains $e_2$ and $e_4$.}
in the drawing $\Gamma'$.
Thus, we may construct a bar 0-visibility representation so that for
each inserted path, $a,u,c,v,b$, the visibility
lines realizing the edges of the path are vertically
aligned (Theorem~4.4~\cite{dett-gdavg-99}).
If we remove the bars representing dummy vertices, the visibility
lines become a line of sight between $a$ and $b$ that is crossed only
by the bar representing vertex $c$.  It follows that the bar 0-visibility
representation, after removing all dummy bars, is a
weak bar 1-visibility representation of $G$.  See
Figure~\ref{fig:orientCross}(c).
\end{proof}

\section{WeB1 graphs are Quasi-planar}
\label{se:WeB1_is_QP}

\begin{theorem}
If a graph $G$ is WeB1, then $G$ is quasi-planar.
\label{thm:qp}
\end{theorem}

\begin{proof}
Let $R$ be a weak bar 1-visibility representation of $G=(V,E)$.
We show that the set of \emph{all} edges $E'$ realized by the
representation $R$ (i.e., the strong bar 1-visibility graph of $R$) forms a
quasi-planar graph.
Since $E$ is a subset of $E'$, $G$ is quasi-planar.

We construct a quasi-planar drawing, $Q$, from the bar representation
$R$ as follows.
In $Q$, place vertex $v$ at the left endpoint, $\ell(v)$, of the bar
representing $v$ in $R$.
The edges of $E'$ are in one of two classes.  Let $E'_0 \subseteq E'$
be the edges, called \emph{blue edges}, realized in $R$ by a direct
visibility between bars.  Let $E'_1 = E' - E'_0$ be
the remaining edges of $G'$, called \emph{red edges}, that is, those
that are only realized by a visibility through another bar.
For a blue edge $(u,v)$, with bar $u$ below bar $v$, draw a
polygonal curve in $Q$ consisting of three segments:
the middle segment is nearly identical to the
\emph{rightmost} vertical visibility
segment that connects bar $u$ with bar $v$,
but it
starts $\gamma$ (a small, positive value) above bar $u$,
ends $\gamma$ below bar $v$, and
is shifted $\gamma$ to the left.
The first and third segments connect $\ell(u)$
to the bottom of the middle segment and the top of the middle segment
to $\ell(v)$, respectively.
We choose $\gamma$ to be smaller than half the minimum positive
difference between bar $x$-coordinates and bar $y$-coordinates,
so a vertical middle segment from one edge does not intersect a (nearly)
horizontal first or third segment from another edge,
and a (nearly) horizontal segment from one edge does not intersect a
(nearly) horizontal segment from another edge.
Thus the curves representing blue edges do not cross.

\begin{figure}
\begin{center}
\includegraphics{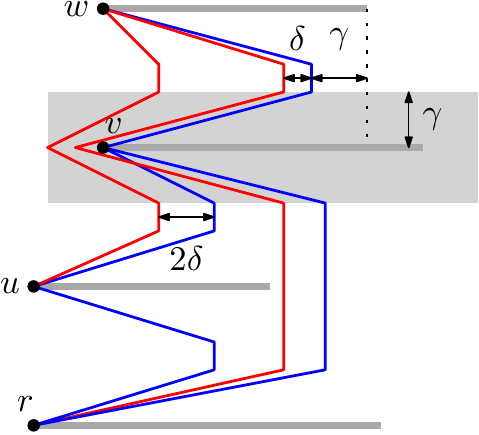}
\end{center}
\caption{Construction of the quasi-planar drawing $Q$.  The shaded
  region around bar $v$ contains no vertical edge segments.
  The values of
  $\gamma$ and $\delta$ in the figure are larger than what they would
  be in a true construction. }
\label{fig:routeRedBlue}
\end{figure}

For a red edge $(u,w)$, let $v$ be the bar that is crossed by the
\emph{rightmost} 1-visibility segment, $\sigma$, that connects bar $u$
with bar $w$.
We call $v$ the \emph{bypass vertex} for the red edge $(u,w)$.
Draw edge $(u,w)$ as a polygonal curve in $Q$ consisting of
six segments:
the first three connect $\ell(u)$ to $\ell(v)$ (as above) where the
middle segment lies $\gamma$ to the left of $\sigma$, and the last
three connect $\ell(v)$ to $\ell(w)$ (as above) where, again, the middle
segment lies $\gamma$ to the left of $\sigma$.
The edges $(u,v)$ and $(v,w)$ are in $E'_0$ and therefore have
polygonal curves in $Q$ that lie on or to the right of the curve for
$(u,w)$.  In order to prevent the curve for $(u,w)$ from intersecting
the curves for $(u,v)$ and $(v,w)$ (except at $\ell(u)$ and
$\ell(w)$), we shift all the points of the curve for $(u,w)$, except
$\ell(u)$ and $\ell(w)$, slightly to the left.
The amount of this shift depends on the red edges that
have $v$ as a bypass vertex.
If $k$ red edges with bypass vertex $v$ have 1-visibility segments to
the right of $\sigma$ then the shift is by $(k+1) \delta$, where
$\delta$ is a positive value that is smaller than $\gamma/|E'|^2$.
In this way, no two red edges with the same bypass vertex intersect,
and no two red edges that share an endpoint intersect.

Note that no vertical edge segments intersect the interior of the
region that is $L_{\infty}$-distance $\gamma$ from a bar, and all
(nearly) horizontal edge segments lie in such a region for some bar.
Thus all edge curve intersections occur within such regions.
See Figure~\ref{fig:routeRedBlue}.

Suppose that the drawing $Q$ is not quasi-planar.
Consider a triple of edges (edge curves) that mutually intersect in $Q$.
We claim that exactly one of these edges is blue.
Since no two blue edges intersect, at most one edge in the triple is blue.
Also, three red edges cannot mutually intersect 
since these edges can only intersect near one of their bypass vertices, call it $v$.
If two red edges share $v$ as a bypass vertex then they do not intersect.
Thus, two of the three red edges have $v$ as an endpoint and therefore
those two don't intersect.

Let $uv$ be the one blue edge in the triple of mutually
intersecting edges.
The intersection of a blue edge and a red edge must
occur near the bypass vertex of the red edge.
Since both red edges in the triple intersect edge $uv$,
they must have bypass vertices $u$ or $v$.
Because they intersect, they cannot share the same bypass vertex, and
one must have an endpoint at $u$ and the other an endpoint at $v$.
Thus the curves representing both red edges have three segments from
$u$ to $v$ and these segments lie to the left of the curve
representing the blue edge $uv$.
Thus neither intersects the blue edge, which is a contradiction.
\end{proof}

\section{Squares of planar 1-flow networks are WeB1}
\label{se:F1sq_is_WeB1}

An acyclic digraph is called \emph{upward planar} if it admits a planar drawing
where all edges are represented by curves monotonically increasing in a common direction.
An upward planar digraph with one source $s$ and one sink $t$, embedded so that $s$ and $t$ are
on the outer face, is called \emph{planar $st$-digraph}.

For a planar $st$-digraph $G = (V,E)$,
let $\leftw(v)$ (resp. $\rightw(v)$) denote the face of $G$ separating the
incoming from the outgoing edges in clockwise (resp. counterclockwise) order.
A \emph{topological numbering} of $G$ is an assignment of numbers to the
vertices of $G$, such that for every edge $(u,v)$, the number assigned
to $v$ is greater than the number assigned to $u$. The numbering is \emph{optimal} if the range of the
numbers assigned to the vertices is minimized.

Recall that a \emph{planar} \emph{$k$-flow network} is an upward planar 
digraph in which every vertex $v$ has $\min\{\indeg(v),
\outdeg(v)\} \leq k$. Recall also that $k$-flow$^2$ denote the class of graphs that are the squares of planar $k$-flow networks. 

As we already mentioned, a bar layout that represents a bar $0$-visibility graph $G$
also represents a family of weak bar $1$-visibility graphs each of
which is a spanning subgraph of $G^2$. In other words, every weak bar $1$-visibility graph is a spanning subgraph of the
square of a bar $0$-visibility graph. In the following we investigate the reverse question, thus, we investigate
which bar $0$-visibility graphs have squares that are weak bar $1$-visible.



\begin{theorem}
\label{th:1-flow}
The square of a planar 1-flow network is WeB1.
\end{theorem}
\begin{proof}
Let $G'$ be a planar 1-flow network and $G$ be a planar $st$-digraph for which $G'$
is a spanning subgraph.
We will prove in Lemma~\ref{lemma:st_augmentation} that such $G$ exists.
The argument is a slight modification of the 
method used to prove Theorem~6.1~\cite{dett-gdavg-99}.

\begin{lemma}
\label{lemma:st_augmentation}
Any planar 1-flow network is a spanning subgraph of an $st$-digraph that is
also a 1-flow network. 
\end{lemma}
\begin{proof}
Let $G'$ be a 1-flow network, i.e., an upward planar digraph with
$\min \{\indeg(v), \outdeg(v) \} \leq 1$, for each vertex $v$.  We add
edges to $G'$ to make it a planar 1-flow network $G$, with a unique source
and a unique sink.  
For an upward planar drawing $\Gamma'$ of $G'$, let
$t_1,\dots,t_k$ (resp. $s_1,\dots,s_f$) be the sinks (resp. sources)
of $G'$ that are on the outer face, where $t_1$ (resp. $s_1$) has the
largest (resp. smallest) $y$-coordinate (see Figure~\ref{fig:cancel_sinks}). 
Add an edge from each of $t_2, \dots, t_k$ to $t_1$ and from $s_1$ to
each of $s_2, \dots, s_f$ so that the resulting drawing $\Gamma''$ is
planar.
Call the new planar 1-flow network $G''$.

\begin{figure}
\centering
\includegraphics[scale=0.7]{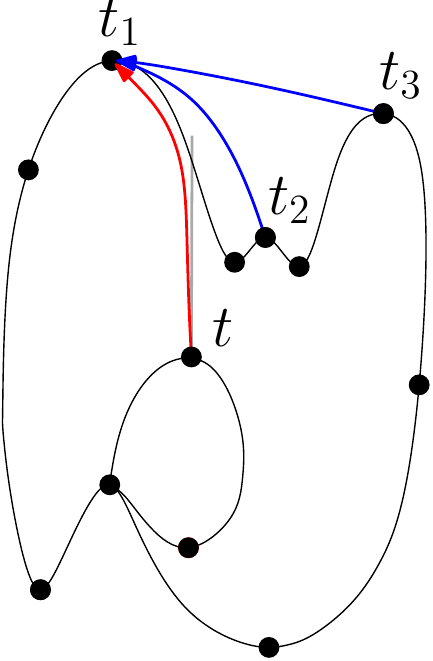}
\caption{Illustration for the proof of Lemma~\ref{lemma:st_augmentation}. Blue edges cancel sinks on the outer face. Red edge cancels a sink of an inner face.}
\label{fig:cancel_sinks}
\end{figure}

Let $t$ be a sink of $G''$.
Consider a vertical half-line $\ell$, originating at $t$ to
$+\infty$. If $t\neq t_1$, half-line $\ell$ crosses a boundary of an
interior face $f$ of $\Gamma''$ that contains $t$,
since otherwise $t$ would have been on the outer face of $\Gamma'$ and
would not be a sink in $G''$ (the edge $(t, t_1)$ would be in $G''$).
We follow half-line $\ell$ and the boundary of face $f$ upward until we reach a sink $t'$ of the
face and add an edge $(t,t')$ to $G''$. Vertex $t'$
either has no outgoing edge, i.e., is a sink of $G''$, or already has
two incoming edges. Thus, the addition of $(t,t')$ keeps $G''$ a
1-flow network. Moreover, edge $(t,t')$ does not create any crossing and keeps the graph upward, therefore
after this step $G''$ is still a planar 1-flow network. The step cancels a sink of $G''$. We repeat this step
until no other sink except for $t_1$ remains.
We perform a symmetric procedure for the remaining sources. The
resulting graph $G$ is a planar 1-flow network. Since only edges have
been added, $G'$ is a spanning subgraph of $G$.
\end{proof}

\begin{figure}
\centering
\subfigure[]{\label{fig:TTcase1}\includegraphics[scale=1]{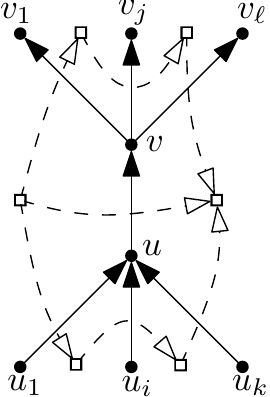}}
\hfill
\subfigure[]{\label{fig:TTcase2}\includegraphics[scale=1]{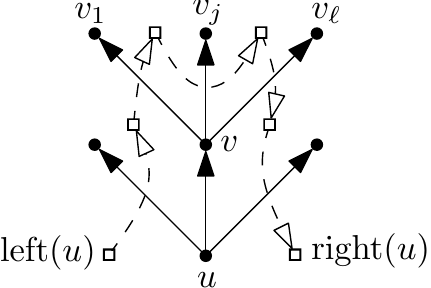}}
\hfill
\subfigure{\label{fig:TT}\includegraphics[scale=1]{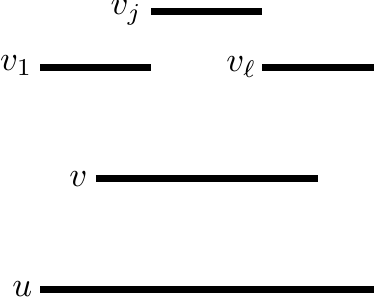}}
\caption{Illustration for the proof of Theorem~\ref{th:1-flow}}
\label{fig:TT_algorithm}
\end{figure}

We come back to the proof of the theorem. In the following we show that the bar 0-visibility representation
$\Gamma$ of $G$ produced by the algorithm of
Tamassia and Tollis~\cite{tt86} is a WeB1 visibility representation of
$G^2$. Since $G'$ is a spanning subgraph of $G$, $G'^2$ is a spanning subgraph of $G^2$, and therefore $\Gamma$ is a WeB1 visibility representation of $G'^2$.
We first review the construction of $\Gamma$.
Let $G^\star$ be the dual of $G$, where each of $G^\star$ is directed so that it crosses the corresponding edge of $G$ from its left to its right. It is easy to see that $G^\star$ is a planar $st$-digraph~\cite{dett-gdavg-99}. Let $\psi$ and $\chi$ be the functions
that assign an optimal topological numbering to the vertices
of $G$ and $G^\star$, respectively. In $\Gamma$, vertex $v$ is
represented as a horizontal bar at $y$-coordinate $\psi(v)$ and with
end-points at $x$-coordinates $\chi(\leftw(v))$ and
$\chi(\rightw(v))-1$. We show that each edge of $G^2$ of the form
$(u,w)$, such that $(u,v), (v,w)\in G$, exists in $\Gamma$ and is represented
by a vertical line crossing only one vertex $v$.
Assume that $v$ has one incoming and several outgoing edges. The case
when $v$ has one outgoing and several incoming edges can be proven
symmetrically. Let $(u,v)$ be the only incoming edge of $v$.
If edge $(u,v)$ is the only outgoing edge of $u$
(Figure~\ref{fig:TTcase1}), $\chi(\leftw(u))=\chi(\leftw(v))$ and
$\chi(\rightw(u))=\chi(\rightw(v))$. Therefore $u$ and $v$ are
represented in $\Gamma$ as two bars with the same left and right
ends. If vertex $u$ has more outgoing edges
(Figure~\ref{fig:TTcase2}), $\chi(\leftw(u))<\chi(\leftw(v))$ and
$\chi(\rightw(v))< \chi(\rightw(u))$. Thus generally it holds that
$\chi(\leftw(u)) \leq \chi(\leftw(v))$ and $\chi(\rightw(v)) \leq
\chi(\rightw(u))$ (see Figure~\ref{fig:TT}) and any vertical line that
intersects bar $v$ also intersects bar $u$.
Thus, if $v_1,\dots,v_{\ell}$ are the remaining neighbors of $v$, any
vertical line that represents an edge from $v$ to $v_j$, also
crosses $u$, for any $1 \leq j \leq \ell$. It remains to show that
there is no bar in $\Gamma$ between $u$ and $v$ crossed by such a vertical
line. Let $w$ be a vertex different from $u$ and $v$. By
Lemma~4.3~\cite{dett-gdavg-99},
exactly one of the following directed paths exists:
\begin{inparaenum}[(1)]
\item from $v$ to $w$ in $G$,
\item from $w$ to $v$ in $G$,
\item from $\rightw(v)$ to $\leftw(w)$ in $G^\star$, or
\item from $\rightw(w)$ to $\leftw(v)$ in $G^\star$.
\end{inparaenum}
The first case implies that $\psi(v) < \psi(w)$ and therefore $w$ is
above $v$ in $\Gamma$. The second case implies that the path from $w$
to $v$ passes through $u$, since $(u,v)$ is the only incoming edge to
$v$. Therefore $\psi(w) < \psi(u)$ and $w$ lies below $u$. In the
third case, $\chi(\rightw(v)) < \chi(\leftw(w))$ and, in the  fourth
case, $\chi(\rightw(w))< \chi(\leftw(v))$. Thus, there is no vertex
$w$, that prevents edges $(u,v_j)$, $1\leq j \leq \ell$, to exist in
$\Gamma$.
\end{proof}


\subsection{Limitations on the squares of planar 2-flow networks}

We show that
while the squares of planar 1-flow networks are WeB1, the squares of
some planar 2-flow networks are not.

\begin{theorem}
\label{th:2-flow}
There exists a planar 2-flow network whose square is not WeB1.
\end{theorem}

\begin{proof}
Consider the graph $G$ of Figure~\ref{fig:Flow3_example} oriented
upward. It consists of a $\sqrt{n} \times \sqrt{n}$ grid, rotated by
$45^\circ$. The diagonals are present only in odd rows. Thus, $G$ is a
2-flow network.
Each vertex has out-degree in $G^2$ indicated by its label
in Figure~\ref{fig:Flow3_example}.
Consider the $(\sqrt{n}-2)^2$ vertices that are distance at
least two from the upper boundary vertices in $G$.
At least half of these vertices have out-degree 7 and the others have
out-degree 6.
Thus $G^2$ has more than $\frac{13}{2}(\sqrt{n} - 2)^2$ edges, which
exceeds the upper bound of $6n-20$ on the number of edges in a
WeB1 graph~\cite{degl7}, for sufficiently large $n$.
\end{proof}

\begin{figure}
\centering
\includegraphics[scale=0.7]{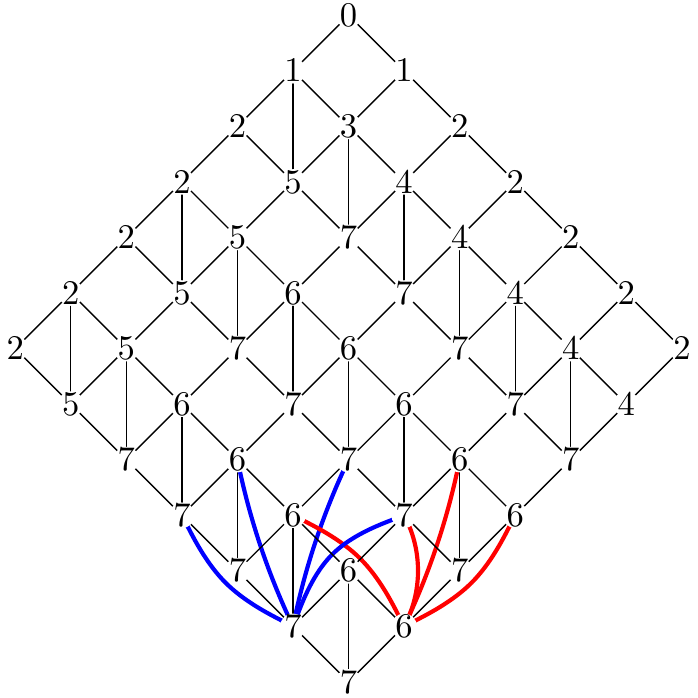}
\caption{Illustration for the proof of Theorem~\ref{th:2-flow}.}
\label{fig:Flow3_example}
\end{figure}

\subsection{Examples for different graph classes related to squares of planar 1-flow networks}

The following two lemmata  introduce examples of
graphs that distinguish certain graph classes in Figure~\ref{fig:world}.

\begin{lemma}
\label{le:K8_not_1flow2}
$K_8$ is not the square of a 1-flow network.
\end{lemma}

\begin{proof}
Suppose $G=(V,E)$ is a 1-flow network such that $G^2 = K_8$.
First, if we view $G$ as a partial order, $\po$, it must be a total
order otherwise two vertices $u \not\po v$ would not be connected in
$G^2$.
We number the vertices $v_1 \po v_2 \po \dots \po v_8$ according to
the total order so that $(v_i, v_{i+1}) \in E$, for all $1 \leq i \leq
8$.
If there exists $3 \leq i \leq 6$ such that $\indeg(v_i) =
\indeg(v_{i+1}) = 1$ then $G^2$ cannot contain the edge $(v_1,
v_{i+1})$.
If there exists $3 \leq i \leq 6$ such that $\outdeg(v_i) =
\outdeg(v_{i+1}) = 1$ then $G^2$ cannot contain the edge $(v_i,
v_8)$.
Also if $\outdeg(v_3) = 1$  and $\indeg(v_6)=1$ then $G^2$ cannot
contain the edge $(v_3, v_6)$.
So $\indeg(v_2)=\indeg(v_3)=\indeg(v_5)=1$ and
$\outdeg(v_4)=\outdeg(v_6)=\outdeg(v_7)=1$.
Thus
$(v_1,v_5) \in G^2$ implies $(v_1,v_4) \in G$;
$(v_2,v_5) \in G^2$ implies $(v_2,v_4) \in G$;
$(v_4,v_7) \in G^2$ implies $(v_5,v_7) \in G$;
$(v_4,v_8) \in G^2$ implies $(v_5,v_8) \in G$;
$(v_3,v_6) \in G^2$ implies $(v_3,v_6) \in G$;
and
$(v_1,v_6) \in G^2$ implies $(v_1,v_6) \in G$.
Also
$(v_3,v_8) \in G^2$ implies $(v_3,v_7) \in G$ or $(v_3,v_8) \in G$;
and
$(v_1,v_8) \in G^2$ implies $(v_1,v_7) \in G$ or $(v_1,v_8) \in G$.
Each of these four possibilities yields a non-planar $G$ since in each
case $\{v_1,v_3,v_5\}$ and either $\{v_4,v_6,v_8\}$ or $\{v_4,v_6,v_7\}$
form a subdivision of $K_{3,3}$ in $G$.
\end{proof}

Let $S_3$ denote the graph consisting of a cycle of length $6$ with an inscribed triangle (Figure~\ref{fig:s3}.a).

\begin{lemma}
\label{lemma_s3}
$S_3$ is a planar StB1 graph and is not the square of a 1-flow network.
\end{lemma}

\begin{proof}

\begin{figure}
\centering
\subfigure[]{\label{fig:s_3}\includegraphics{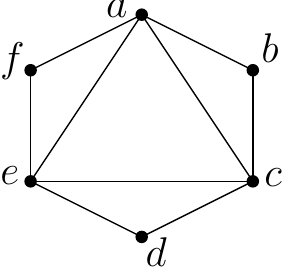}}
\hspace{2cm}
\subfigure[]{\label{fig:example_s_3}\includegraphics{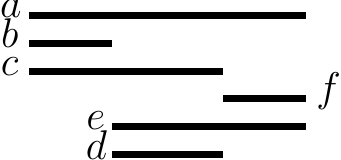}}
\caption{(a) Graph $S_3$ of Lemma~\ref{lemma_s3}.  (b) StB1 representation of $S_3$.}
\label{fig:s3}
\end{figure}

A StB1 representation of $S_3$ is shown in
Figure~\ref{fig:s3}(b). In the following we show that there
exists no 1-flow network $G$, such that $G^2=S_3$. We denote by $ab$ the
undirected edge between $a$ and $b$, and by $(a,b)$ the directed
edge from $a$ to $b$.
For the sake of contradiction assume such $G$ exists. We first assume
that $G$ does not contain all the edges of the external face of
$S_3$. Without loss of generality assume that $ab$ is not in $G$. Then
both $bc$ and $ac$ must be in $G$. Moreover they must be similarly
directed. Assume that they are directed as $(b,c)$ and $(c,a)$
($(c,b)$ and $(a,c)$, respectively). Then edge $dc$ is not in $G$,
since $(d,c)$ would induce $(d,a)$ (resp. $(d,b)$) in $G^2$,
while $(c,d)$ would induce $(b,d)$ (resp. $(a,d)$).  Thus both edges
$ec$ and $ed$ must be in $G$. Edge $ec$ must be oriented as
$(e,c)$ (resp. $(c,e)$), otherwise edge $(b,e)$ (resp. $(e,b)$) is in
$G^2$. Thus, $(d,e) \in G$ (resp. $(e,d) \in G$). Similarly, we
conclude that $(a,e) \in G$ (resp. $(e,a) \in G$), and therefore we
get a cycle $ace$ in $G$, which is a contradiction to the upward
condition of 1-flow networks.

Now, assume that $G$ contains all the edges of the outer face. We
distinguish cases based on the length of the directed paths
contained in the outer face.
If the longest path has length one then none of the edges
$ae$, $ac$, $ec$ are induced in $G^2$ by outer edge paths, and so at
least one must be in $G$.
But, any orientation of this edge creates
an additional edge in $G^2$, which does not belong to $S_3$.

If there exists a path of length three we get a contradiction, since
one of its length two subpaths induces an edge not in $S_3$.

Assume there exists a single path of length two, and no path of length
three. Then the middle vertex of the path must be $b$, $d$, or
$f$,
otherwise the path induces an edge not in $S_3$. Without loss of
generality assume that the path is $(a,b,c)$. Then $fa$ is
oriented as $(a,f)$ and $dc$ as  $(d,c)$. Any orientation of $fe$ and
$ed$ either introduces a path of length three (above case) or two
paths of length two (the next case).

Finally, assume there are two paths of length two. They must share a
vertex, otherwise one of them induces an edge not in $S_3$, and they
must be oriented opposite, otherwise a path of length three exists.
Without loss of generality we can assume that they are either paths
$(e,f,a)$ and $(c,b,a)$,  or paths $(a,f,e)$ and
$(a,b,c)$.  In case of $(e,f,a)$ and $(c,b,a)$,
edges $ed$ and $cd$ must be oriented as $(e,d)$ and $(c,d)$. Thus edge
$ec$ must be in $G$. But any orientation of $ec$ induces an edge in
$G^2$ that is not in $S_3$. Similarly with paths $(a,f,e)$ and
$(a,b,c)$.

\end{proof}

\section{Conclusion and Open Problems}

In this paper we investigated the relation of bar $1$-visibility graphs with other classes of graphs that are ``close to planar'', by proving that: (i) All $1$-planar graphs are WeB1, $(ii)$ All WeB1 graphs are quasi-planar, and finally that $(iii)$ All $1$-flow$^2$ graphs are WeB1, however not all $2$-flow$^2$ graphs are WeB1.
While these results provide some insight on the class of bar $1$-visibility graphs it would be interesting to provide  
a complete characterization of WeB1 or StB1 graphs.
Regarding the relation of WeB1 and $k$-flow$^2$ graphs, what can we say about the squares of planar digraphs, where for each
vertex $v$, either $\min\{\indeg(v), \outdeg(v)\}=1$, or $\indeg(v) =
\outdeg(v) =2$ (except for $v =s,t$)?


\bibliography{gdbiblio,extra,extra1}

\begin{thebibliography}{10}

\bibitem{AckermanE09}
E.~Ackerman.
\newblock On the maximum number of edges in topological graphs with no four
  pairwise crossing edges.
\newblock {\em Discrete \& Computational Geometry}, 41(3):365--375, 2009.

\bibitem{Ackerman2007563}
E.~Ackerman and G.~Tardos.
\newblock On the maximum number of edges in quasi-planar graphs.
\newblock {\em Journal of Combinatorial Theory, Series A}, 114(3):563 -- 571,
  2007.

\bibitem{DBLP:conf/gd/AuerBGH12}
C.~Auer, F.-J. Brandenburg, A.~Glei{\ss}ner, and K.~Hanauer.
\newblock On sparse maximal 2-planar graphs.
\newblock In Didimo and Patrignani \cite{DBLP:conf/gd/2012}, pages 555--556.

\bibitem{BessyHB06}
S.~Bessy, F.~Havet, and E.~Birmel{\'e}.
\newblock Arc-chromatic number of digraphs in which every vertex has bounded
  outdegree or bounded indegree.
\newblock {\em J. Graph Theory}, 53(4):315--332, 2006.

\bibitem{DBLP:journals/dam/BorodinKRS01}
O.~V. Borodin, A.~V. Kostochka, A.~Raspaud, and E.~Sopena.
\newblock Acyclic colouring of 1-planar graphs.
\newblock {\em Discrete Applied Mathematics}, 114(1-3):29--41, 2001.

\bibitem{Brandenburg13}
F.~Brandenburg.
\newblock 1-visibility representations of 1-planar graphs.
\newblock {\em CoRR}, abs/1308.5079, 2013.

\bibitem{DBLP:conf/gd/BrandenburgEGGHR12}
F.-J. Brandenburg, D.~Eppstein, A.~Glei{\ss}ner, M.~T. Goodrich, K.~Hanauer,
  and J.~Reislhuber.
\newblock On the density of maximal 1-planar graphs.
\newblock In Didimo and Patrignani \cite{DBLP:conf/gd/2012}, pages 327--338.

\bibitem{DBLP:journals/dam/CzapH12}
J.~Czap and D.~Hud{\'a}k.
\newblock 1-planarity of complete multipartite graphs.
\newblock {\em Discrete Applied Mathematics}, 160(4-5):505--512, 2012.

\bibitem{degl7}
A.~Dean, W.~Evans, E.~Gethner, J.~Laison, M.~Safari, and W.~Trotter.
\newblock Bar k-visibility graphs.
\newblock {\em Journal of Graph Algorithms and Applications}, 11(1):45--59,
  2007.

\bibitem{dett-gdavg-99}
G.~{Di Battista}, P.~Eades, R.~Tamassia, and I.~G. Tollis.
\newblock {\em Graph Drawing -- Algorithms for the Visualization of Graphs}.
\newblock Prentice Hall, 1999.

\bibitem{dl-cargd-12}
W.~Didimo and G.~Liotta.
\newblock The crossing angle resolution in graph drawing.
\newblock In J.~Pach, editor, {\em Thirty Essays on Geometric Graph Theory}.
  Springer, 2012.

\bibitem{DBLP:conf/gd/2012}
W.~Didimo and M.~Patrignani, editors.
\newblock {\em Graph Drawing - 20th International Symposium, GD 2012, Redmond,
  WA, USA, September 19-21, 2012, Revised Selected Papers}, volume 7704 of {\em
  Lecture Notes in Computer Science}. Springer, 2013.

\bibitem{DBLP:conf/gd/EadesHKLSS12}
P.~Eades, S.-H. Hong, N.~Katoh, G.~Liotta, P.~Schweitzer, and Y.~Suzuki.
\newblock Testing maximal 1-planarity of graphs with a rotation system in
  linear time - (extended abstract).
\newblock In Didimo and Patrignani \cite{DBLP:conf/gd/2012}, pages 339--345.

\bibitem{DBLP:conf/cocoon/HongELP12}
P.~Eades, S.-H. Hong, G.~Liotta, and S.-H. Poon.
\newblock F{\'a}ry's theorem for 1-planar graphs.
\newblock In J.~Gudmundsson, J.~Mestre, and T.~Viglas, editors, {\em COCOON},
  volume 7434 of {\em Lecture Notes in Computer Science}, pages 335--346.
  Springer, 2012.

\bibitem{DBLP:journals/dam/EadesL13}
P.~Eades and G.~Liotta.
\newblock Right angle crossing graphs and 1-planarity.
\newblock {\em Discrete Applied Mathematics}, 161(7-8):961--969, 2013.

\bibitem{DBLP:journals/dm/FabriciM07}
I.~Fabrici and T.~Madaras.
\newblock The structure of 1-planar graphs.
\newblock {\em Discrete Mathematics}, 307(7-8):854--865, 2007.

\bibitem{DBLP:journals/jgaa/FelsnerM08}
S.~Felsner and M.~Massow.
\newblock Parameters of bar k-visibility graphs.
\newblock {\em J. Graph Algorithms Appl.}, 12(1):5--27, 2008.

\bibitem{GenesonKT13}
J.~Geneson, T.~Khovanova, and J.~Tidor.
\newblock Convex geometric (k+2)-quasiplanar representations of semi-bar
  k-visibility graphs.
\newblock {\em CoRR}, abs/1307.1169, 2013.

\bibitem{DBLP:conf/wg/GiacomoDLM12}
E.~D. Giacomo, W.~Didimo, G.~Liotta, and F.~Montecchiani.
\newblock h-quasi planar drawings of bounded treewidth graphs in linear area.
\newblock In M.~C. Golumbic, M.~Stern, A.~Levy, and G.~Morgenstern, editors,
  {\em WG}, volume 7551 of {\em Lecture Notes in Computer Science}, pages
  91--102. Springer, 2012.

\bibitem{DBLP:journals/dm/Korzhik08a}
V.~P. Korzhik.
\newblock Minimal non-1-planar graphs.
\newblock {\em Discrete Mathematics}, 308(7):1319--1327, 2008.

\bibitem{DBLP:conf/gd/KorzhikM08}
V.~P. Korzhik and B.~Mohar.
\newblock Minimal obstructions for 1-immersions and hardness of 1-planarity
  testing.
\newblock In I.~G. Tollis and M.~Patrignani, editors, {\em Graph Drawing},
  volume 5417 of {\em Lecture Notes in Computer Science}, pages 302--312.
  Springer, 2008.

\bibitem{lec-aptg-67}
A.~Lempel, S.~Even, and I.~Cederbaum.
\newblock An algorithm for planarity testing of graphs.
\newblock In {\em Theory of Graphs: Internat. Symposium (Rome 1966)}, pages
  215--232, New York, 1967. Gordon and Breach.

\bibitem{pachss96}
J.~Pach, F.~Shahrokhi, and M.~Szegedy.
\newblock Applications of the crossing number.
\newblock {\em Algorithmica}, 16(1):111--117, 1996.

\bibitem{PachT97}
J.~Pach and G.~T{\'o}th.
\newblock Graphs drawn with few crossings per edge.
\newblock {\em Combinatorica}, 17(3):427--439, 1997.

\bibitem{DBLP:journals/corr/abs-1302-4870}
S.~Sultana, M.~S. Rahman, A.~Roy, and S.~Tairin.
\newblock Bar 1-visibility drawings of 1-planar graphs.
\newblock {\em CoRR}, abs/1302.4870, 2013.

\bibitem{SultanaRRT13}
S.~Sultana, M.~S. Rahman, A.~Roy, and S.~Tairin.
\newblock Bar 1-visibility drawings of 1-planar graphs.
\newblock In {\em Proc. 1st International Conference on Applied Algorithms
  (ICAA14)}, 2014.
\newblock to appear.

\bibitem{DBLP:journals/dm/Suzuki10}
Y.~Suzuki.
\newblock Optimal 1-planar graphs which triangulate other surfaces.
\newblock {\em Discrete Mathematics}, 310(1):6--11, 2010.

\bibitem{tt86}
R.~Tamassia and I.~G. Tollis.
\newblock A unified approach to visibility representations of planar graphs.
\newblock {\em Discrete and Computational Geometry}, 1(4):321--341, 1986.

\bibitem{tarjan1983data}
R.~Tarjan.
\newblock {\em Data Structures and Network Algorithms}.
\newblock Applied Mathematics Series. Society for Industrial and Applied
  Mathematics, 1983.

\bibitem{th-rdg-88}
C.~Thomassen.
\newblock Rectilinear drawings of graphs.
\newblock {\em Journal of Graph Theory}, 12(3):335--341, 1988.

\bibitem{Wismath85}
S.~Wismath.
\newblock Characterizing bar line-of-sight graphs.
\newblock In {\em Proc. 1st ACM Symp. Comput. Geom.}, pages 147--152. ACM
  Press, 1985.

\end{thebibliography}

\bibliographystyle{abbrv}

\newpage
\appendix

\end{document}